\documentclass[10pt,conference]{IEEEtran}

\usepackage[dvips]{color}
\usepackage{epsf}

\usepackage{hyperref}
\usepackage{float}
\usepackage{subfigure}
\usepackage{bbm}
\usepackage{anyfontsize}
\usepackage{t1enc}
\usepackage{amsmath}
\usepackage{amsxtra}
\usepackage{graphicx}
\usepackage{epsfig}
\usepackage{latexsym}
\usepackage{amsfonts}
\usepackage{lipsum}% http://ctan.org/pkg/lipsum
\usepackage{rawfonts}
\usepackage[latin1]{inputenc}
\usepackage[T1]{fontenc}
\usepackage{calc}
\usepackage{url}
\usepackage{enumerate}
\usepackage{color}
\usepackage{amssymb}
\usepackage{upref}
\usepackage{epic,eepic}
\usepackage{times}
\usepackage{cite}
\usepackage{dsfont}
\usepackage{mathrsfs}
\usepackage{verbatim}
\usepackage{float}
\usepackage{chngcntr}
\usepackage{bbm}
\usepackage{algorithm}% http://ctan.org/pkg/algorithm
\usepackage{algpseudocode}% http://ctan.org/pkg/algorithmicx

\DeclareMathOperator*{\maxim}{maximize}

\newcommand{\normipt}[1]{\left\|#1\right\|}

\newcommand{\card}{\textbf{Card}}

\newtheorem{theorem}{\bf{ Theorem}}
\newtheorem{proposition}{\bf{ Proposition}}

\newtheorem{definition}{\bf{ Definition}}
\newtheorem{remark}{\bf{Remark}}
\newcommand{\qed}{\nobreak \ifvmode \relax \else
  \ifdim\lastskip<1.5em \hskip-\lastskip
  \hskip1.5em plus0em minus0.5em \fi \nobreak
  \vrule height0.75em width0.5em depth0.25em\fi}
\newlength{\totlinewidth}

\newcounter{substep}

  {\end{list}}

\newlength{\aligntop}
\newlength{\alignbot}
 \makeatletter
\renewenvironment{align}{%
  \vspace{\aligntop}
  \start@align\@ne\st@rredfalse\m@ne
}{%
  \math@cr \black@\totwidth@
  \egroup
  \ifingather@
    \restorealignstate@
    \egroup
    \nonumber
    \ifnum0=`{\fi\iffalse}\fi
  \else
    $$%
  \fi
  \ignorespacesafterend%
  \vspace{\alignbot}\par\noindent
} \makeatother

\IEEEoverridecommandlockouts

\begin{document}
% paper title
\title{\huge Downlink Cell Association and Load Balancing for Joint Millimeter Wave-Microwave Cellular Networks}\vspace{0em}
\author{
\authorblockN{Omid Semiari$^{\dag}$, Walid Saad$^{\dag}$, and Mehdi Bennis$^\ddag$}\\\vspace*{-.5em}
\authorblockA{\small $^{\dag}$Wireless@VT, Bradley Department of Electrical and Computer Engineering, Virginia Tech, Blacksburg, VA, USA, \\Emails: \protect\url{{osemiari,walids}@vt.edu}\\
\small $^\ddag$ Centre for Wireless Communications, University of Oulu, Finland, Email: \url{bennis@ee.oulu.fi}
}\vspace*{-3.1em}
   % \thanks{This research was supported by the U.S. National Science Foundation under Grants CNS-1460316 and CNS-1526844 and by the Academy of Finland.}%
  }
%
% make the title area
\maketitle
\begin{abstract}
The integration of millimeter-wave base stations (mmW-BSs) with conventional microwave base stations ($\mu$W-BSs) is a promising solution for enhancing the quality-of-service (QoS) of emerging 5G networks. However, the significant differences in the signal propagation characteristics over the mmW and $\mu$W frequency bands will require novel cell association schemes cognizant of both mmW and $\mu$W systems. In this paper, a novel cell association framework is proposed that considers both the blockage probability and the achievable rate to assign user equipments (UEs) to mmW-BSs or $\mu$W-BSs. The problem is formulated as a one-to-many matching problem with \textit{minimum quota constraints} for the BSs that provides an efficient way to balance the load over the mmW and $\mu$W frequency bands. To solve the problem, a distributed algorithm is proposed that is guaranteed to yield a Pareto optimal and two-sided stable solution. Simulation results show that the proposed matching with minimum quota (MMQ) algorithm outperforms the conventional max-RSSI and max-SINR cell association schemes.  
In addition, it is shown that the proposed MMQ algorithm can effectively balance the number of UEs associated with the $\mu$W-BSs and mmW-BSs and achieve further gains, in terms of the average sum rate.
 \vspace{-0cm}
\end{abstract}
\section{Introduction} \vspace{-0cm}
The integration of cellular networks with millimeter-wave (mmW) communication links is a promising solution to meet the high data traffic requirements of tomorrow's wireless services \cite{6736746,Rangan14,andrew16,park15,Ghosh14}. However, mmW communication is known to be inherently intermittent, due to the susceptibility  of its links to signal blockage, due to shadowing by human, buildings, and other obstacles. To this end, mmW base stations (mmW-BSs) must coexist with the conventional microwave base stations ($\mu$W-BSs) to provide $\mu$W connectivity for users, when a reliable mmW communication is not feasible \cite{andrew16,park15}.

Such integrated mmW-$\mu$W networks introduce new challenges for cellular resource management. In particular, the association of user equipments (UEs) to the BSs must now account for the presence of two radio access technologies (RATs) with significantly different propagation environments. In fact, conventional approaches such as maximum signal-to-interference-plus-noise-ratio (max-SINR) and maximum signal strength indicator (max-RSSI) may result in significantly unbalanced load distributions and may not be directly applicable to the multi-RAT setting. That is due to three key reasons: a) mmW links are highly intermittent and have a higher path loss than $\mu$W, b) mmW communication is mostly limited by noise rather than interference, and c) more bandwidth is available at mmW band compared to the $\mu$W frequency band.

The problem of cell association with load balancing has been extensively studied in heterogeneous cellular networks \cite{6008531,6497017,4607241,6774981,andrew16,park15}. The work in \cite{6008531} studies the performance of the max-SINR cell association for heterogeneous networks (HetNets) with load balancing via cell range expansion (CRE). The authors in \cite{6497017} propose a cell association approach based on convex optimization to find a load-aware distributed cell association algorithm for HetNets. Moreover, in \cite{4607241}, a game-theoretic approach is adopted for network selection in HetNets, using an evolutionary game approach. For mmW networks, the work in \cite{6774981} presents a distributed algorithm that yields a fair cell association. A stochastic geometry framework is used in \cite{andrew16} for the decoupled uplink-downlink cell association for traditional macrocells and mmW small cell networks. In addition, the authors in \cite{park15} study resource allocation for mmW-$\mu$W networks where cell association is decoupled in the uplink for mmW users.   

The existing works in \cite{6008531,6497017,4607241,6774981} have focused on $\mu$W or mmW networks, separately and in isolation, and thus, they cannot be applied to integrated mmW-$\mu$W cellular networks. In addition, the authors in \cite{andrew16} and \cite{park15} consider max-RSSI cell association. However, max-RSSI is not a proper association metric for integrated mmW-$\mu$W networks, since it does not properly reflect the achievable rate of the users. Indeed, this rate depends on the allocated bandwidth and the interference, which are completely different between mmW and $\mu$W.

The main contribution of this paper is to introduce a novel cell association framework with load balancing for integrated mmW-$\mu$W cellular networks. First, we show that conventional max-SINR and max-RSSI cell associations can result in significant unbalanced load in mmW and $\mu$W networks. Then, we formulate the proposed cell association problem as a \textit{matching game with minimum quota constraints}, in which the BSs can adjust their minimum quota, in terms of the number of UEs they serve, to balance the network's load. For this game, we show that classical matching solutions such as in \cite{Roth92} and \cite{eduard11} cannot be applied. In contrast, to solve our problem, we propose a novel distributed algorithm that allows UEs to submit association requests to either the mmW-BS or $\mu$W-BS that maximizes its average achievable rate. To achieve a balanced load, BSs approve UEs' requests such that the quota constraints are met. We show that the proposed algorithm yields a Pareto optimal (PO) and stable solution for the UEs. Simulation results show the effectiveness of our approach in integrated mmW-$\mu$W networks.

The rest of this paper is organized as follows. Section II presents the problem formulation. Section III formulates the problem as a matching game. Section IV presents the proposed algorithm. Simulation results are analyzed in Section V. Section VI concludes the paper.

\section{System Model}
Consider the downlink of a cellular network, composed of a set $\mathcal{N}_1$ of $N_1$ mmW-BSs and a set $\mathcal{N}_2$ of $N_2$ $\mu$W-BSs. In this network, a set $\mathcal{M}$ of $M$ UEs are deployed and must be assigned to one mmW-BS or $\mu$W-BS. UEs and BSs are distributed uniformly and randomly within a planar area with radius $r$ centered at $(0,0) \in \mathbb{R}^2$. UEs are equipped with both mmW and $\mu$W RF interfaces allowing them to manage their traffic at both frequency bands.  
\begin{figure}
\centering
\centerline{\includegraphics[width=9cm]{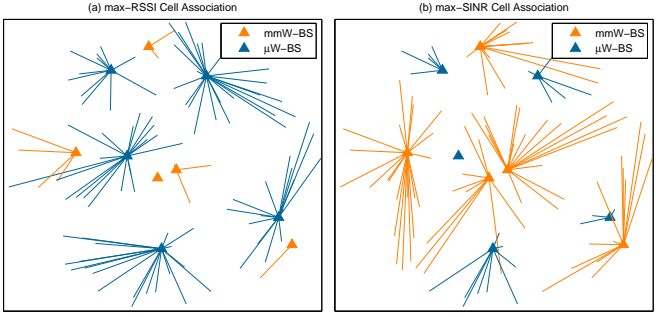}}\vspace{-0.4cm}
\caption{Cell association for an integrated $\mu$W-mmW network using (a) max-RSSI, and (b) max-SINR approaches. The triangles show the BSs and the orange and blue colors represent, respectively, the mmW and $\mu$W links.}\vspace{-.3cm}
\label{model}
\end{figure}
\subsection{Propagation Model at mmW and $\mu$W Frequency Bands}
 
Each mmW link between mmW-BS $n\in \mathcal{N}_1$ and UE $m \in \mathcal{M}$, located at $\boldsymbol{y}_n \in \mathbb{R}^2$ and $\boldsymbol{y}_m \in \mathbb{R}^2$, respectively, is characterized by the transmit power $p_n$, channel gain $g(\boldsymbol{y}_m,\boldsymbol{y}_n)$ and the antenna gain $\psi(\boldsymbol{y}_m,\boldsymbol{y}_n)$. Assuming that the total power $p_n$ is distributed uniformly over the mmW bandwidth, the achievable rate per unit of bandwidth for a UE $m$ assigned to mmW-BS $n$ is given by:
\begin{align}\label{ratemmw}
c_{m,n}^{\text{mmW}}
=\log_2 \!\left(\!1 + \frac{p_n \psi(\boldsymbol{y}_m,\boldsymbol{y}_n)g(\boldsymbol{y}_m,\boldsymbol{y}_n)}{w_1 N_0}\right),
\end{align}
where $w_1$ is the mmW bandwidth, $g(\boldsymbol{y}_m,\boldsymbol{y}_n)$ is the link channel gain, and $N_0$ is the noise power spectral density. Hereinafter, we represent $c_{m,n}^{\text{mmW}}$ by $c_{m,n}^{\text{LoS}}$ and $c_{m,n}^{\text{NLoS}}$, respectively, if the link is line-of-sight (LoS) and non-line-of-sight (NLoS).  
Here, $g(\boldsymbol{y}_m,\boldsymbol{y}_n) = L(\boldsymbol{y}_m,\boldsymbol{y}_n)^{-1}$, where the path loss $L(\boldsymbol{y}_m,\boldsymbol{y}_n)$ in dB follows the model of \cite{Ghosh14}:
 \begin{align}\label{mwpathloss}
 L(\boldsymbol{y}_m,\boldsymbol{y}_n)=b_1+a_1 10\log_{10}(\normipt{\boldsymbol{y}_m-\boldsymbol{y}_n})+\chi,
 \end{align}
where $a_1$ represents the slope of the best linear fit to the propagation measurement in mmW frequency band and $b_1$ is the path loss (in dB) for $1$ meter of distance. In addition, $\chi$ models the deviation in fitting (in dB) which is a Gaussian random variable with zero mean and variance $\xi_1^2$. 

For each UE-BS pair $(m,n)$, let $\zeta_{m,n}$ be a Bernoulli random variable with success probability $\rho_{m,n}$ that indicates whether the mmW link is LoS, $\zeta_{m,n}=1$, or NLoS, $\zeta_{m,n}=0$. Different path loss parameters in \eqref{mwpathloss} are considered for the LoS and NLoS links, as listed in Table. \ref{tabsim}.

At $\mu$W band, the achievable rate per unit of bandwidth for a UE $m \in \mathcal{M}$ associated with $\mu$W-BS $n \in \mathcal{N}_2$ is given by:
\begin{align}\label{ratemuw}
c^{\mu\text{W}}_{m,n}
=\log_2 \left(1 + \frac{p_n g(\boldsymbol{y}_m,\boldsymbol{y}_n)}{\sum_{n'\neq n} p_{n'} g(\boldsymbol{y}_m,\boldsymbol{y}_{n'}) + w_2 N_0}\right),
\end{align}
where the total power $p_n$ is distributed uniformly over the $\mu$W bandwidth, $w_2$, and the channel gain is characterized by parameters, $a_2,b_2$ and $\xi_2$, similar to \eqref{mwpathloss}.
\subsection{Problem Formulation}
The cell association problem can be defined as a decision policy $\pi$ which, for any UE-BS pair $(m,n)$, it outputs a binary variable $x_{m,n} \in \{0,1\}$, where $x_{m,n}=1$ indicates that UE $m$ is assigned to BS $n$, otherwise, $x_{m,n}=0$. 
%Conventional cell association approaches include max-SINR and max-RSSI schemes \cite{6008531,6497017}. Examples of cell association in a joint mmW-$\mu$W network is shown for the max-RSSI and the max-SINR approaches, respectively, in Fig.\ref{model} (a) and Fig.\ref{model} (b).
%The max-SINR cell association policy $\pi^{\textrm{SINR}}$ assigns each UE $m$ to a BS $n = \argmax_{n} \gamma_{m,n}$, where $\gamma_{m,n}$ is the SINR at UE $m$ assigned to BS $n$. 
Further, we define the BS $n$'s load, $\kappa_n$ as
\begin{align}\label{loaddef}
\kappa_n = \sum_{m=1}^{M}x_{m,n}.
\end{align}
Using \eqref{loaddef}, the \emph{maximum load difference} can be defined as the difference of the load for the BSs with the maximum and minimum number of associated UEs, as follow:
\begin{align}\label{loaddifference}
\Delta_{\kappa}(\pi) = \max(\kappa_n)-\min(\kappa_n).
\end{align}

In \eqref{loaddifference}, a smaller $\Delta_{\kappa}(\pi)$ implies better load balancing. In general, it is desirable to achieve uniform loads for all BSs, i.e., $\Delta_{\kappa}(\pi)=0$. However, by using conventional cell association approaches, such as max-SINR and max-RSSI schemes \cite{6008531,6497017}, as shown in Fig. \ref{model}, the network will exhibit a severely unbalanced load. In fact, the max-RSSI scheme assigns most of the UEs to $\mu$W-BS, due to the smaller path loss over the $\mu$W frequency band. On the other hand, the max-SINR scheme assigns most of the UEs to the mmW-BSs, due to the directional transmissions and less interference. 
In addition, in Figs. \ref{sim6} and \ref{simulation7}, we show by simulations that the CRE techniques used in small cell networks \cite{6008531} may not effectively improve load balancing in mmW-$\mu$W networks, due to the large gap in the RSSI and SINR values for mmW and $\mu$W links. Our joint mmW-$\mu$W cell association problem is thus given by:
\begin{IEEEeqnarray}{rCl}\label{eq4}
&&\maxim_{\boldsymbol{x}} \,\, \sum_{n=1}^{N}\sum_{m=1}^{M}x_{m,n}U_{m,n}(\boldsymbol{x}), \IEEEyessubnumber\label{1a}\\
\textrm{s.t.} \,\,\,\,
&& \sum_{n \in \mathcal{N}}x_{m,n}\leq 1,\,\,\,\,\,\,\,\,\,\,\forall m \in \mathcal{M},\IEEEyessubnumber\label{1b}\\
&& \kappa_n \leq q_n^{\textrm{max}},\,\,\,\,\,\,\,\,\,\,\,\,\,\,\,\,\,\,\,\forall n \in \mathcal{N},\IEEEyessubnumber\label{1c}\\
&& \kappa_n \geq q_n^{\textrm{min}},\,\,\,\,\,\,\,\,\,\,\,\,\,\,\,\,\,\,\,\,\forall n \in \mathcal{N},\IEEEyessubnumber\label{1d}\\
&& x_{m,n} \in \{0,1\},\IEEEyessubnumber\label{1e}
\end{IEEEeqnarray}
where $\mathcal{N}=\mathcal{N}_1 \cup \mathcal{N}_2$ is the set of all $N=N_1 + N_2$ BSs and $U_{m,n}$ denotes the utility of the UE $m$ associated to the BS $n$. Moreover, $q_n^{\textrm{max}}$ and $q_n^{\textrm{min}}$ denote, respectively, \emph{the maximum and the minimum quotas} for BS $n$ which represent the maximum and the minimum number of UEs that it can serve. 
We let $0\leq q_n^{\textrm{min}} \leq q_n^{\textrm{max}}$ and $\sum_{n \in \mathcal{N}}q_n^{\textrm{min}} \leq M \leq \sum_{n \in \mathcal{N}}q_n^{\textrm{max}}$ to ensure that a feasible solution exists. As we elaborate later in Section IV, constraints \eqref{1c}-\eqref{1d} are introduced to balance the network's load. Next, we make the following observation:
\begin{remark}
With $q_n^{\textrm{min}}=0$ and $q_n^{\textrm{max}}=M$ for $\forall n \in \mathcal{N}$, the optimization problem \eqref{1a}-\eqref{1e} does not incorporate load balancing.
\end{remark}
The cell association for an arbitrary UE depends on the associations of the other UEs, due to the quota constraints \eqref{1c}-\eqref{1d}. In addition, the utility of a UE may depend on whether the associated BS is a mmW-BS or a $\mu$W-BS. 
%In fact, association metrics such as SINR and path loss can dramatically degrade the performance if the UE-BS link has a low probability of LoS. Therefore, in the next section, we define relevant utility functions for each frequency band, considering blockage at mmW band.
%%%%%%%%%%%%%%%%%%%%%%%%%%%%%%%%%%%%%%%%%%%%%%%%%%%%
\section{Cell Association as a Matching Game with Minimum Quotas}
The downlink association problem in (\ref{1a})-(\ref{1e}) is a $0$-$1$ integer programming problem for assigning UEs to BSs which does not admit a closed-form solution and has exponential complexity\cite{4036195}. In fact, for such a cell association problem, an exhaustive search requires a comparison of $O(N^M)$ assignments, which cannot adapt to the dynamics of dense cellular networks, particularly, when using the mmW frequency band. 

In this regard, centralized cell association schemes require the BSs to send the network information to the radio network controller (RNC). Such implementations carried out by the RNC are updated at relatively long timescales. This can be detrimental for the mmW UEs that frequently experience NLoS transmissions. To this end, we propose a distributed solution for the mmW-$\mu$W cell association problem. 

\subsection{Cell Association as a Matching Game: Preliminaries}
To solve the problem in (\ref{1a})-(\ref{1e}), we propose a novel solution based on \emph{matching theory}, a mathematical framework that provides a decentralized solution with tractable complexity for combinatorial problems, such as the one in (\ref{1a})-(\ref{1e}) \cite{Roth92,eduard11}. A \emph{matching game} is essentially a two-sided assignment problem between two disjoint sets of players in which the players of each set are interested to be matched to the players of the other set, according to \textit{preference relations}. In our model, over each cell association time frame, the set of BSs, $\mathcal{N}$, and the set of UEs, $\mathcal{M}$, are the two sets of players of the matching game. A preference relation $\succ$ is defined as a complete, reflexive, and transitive binary relation between the elements of a given set. Here, we let $\succ_m$ be the preference relation of UE $m$ and denote $n\succ_m n'$, if UE $m$ prefers BS $n$ more than $n'$. Similarly, we use $\succ_n$ to denote the preference relation of BS $n \in \mathcal{N}$.

To define the preference relations, we can introduce individual utility functions for each UE and BS, using which they can rank one another. In the proposed cell association problem, the preference relations of UEs will depend only on the local average achievable rate information, while the BSs will use network-wide information to distribute the loads and maximize the sum utility. In fact, matching-based cell association provides a suitable framework to balance the load by properly adjusting the maximum and minimum BS quotas.

\subsection{Cell Association as a Matching Game}
Each cell association policy $\pi$ determines the allocation of a subset of UEs to each BS. Thus, the problem can be defined as a \textit{one-to-many matching game}:
\begin{definition}\label{def1}
Given two disjoint finite sets of players $\mathcal{M}$ and $\mathcal{N}$, the cell association policy, $\pi$, can be defined as a a one-to-many \textit{matching relation}, $\pi\!:\!\mathcal{N}  \rightarrow \mathcal{M}$ that satisfies 1) $\forall n \in \mathcal{N}, \pi(n) \subseteq \mathcal{M}$, 2) $\forall m \in \mathcal{M}, \pi(m) \in \mathcal{N}$, and 3) $\pi(m)=n$, if and only if  $m \in \pi(n)$.
\end{definition}    
In fact, $\pi(m)=n$ implies that $x_{m,n}=1$, otherwise $x_{m,n}=0$. One can easily see from Definition \ref{def1} that the proposed matching game inherently satisfies the constraints in \eqref{1b} and \eqref{1e}. In addition, $\pi$ is a \textit{feasible matching}, if it satisfies the quota constraints, i.e., $|\pi(m)| \in \{0,1\}$ and $q_n^{\textrm{min}} \leq \kappa_n=|\pi(n)| \leq q_n^{\textrm{max}}$, where $|.|$ denotes the set cardinality. Next, we define suitable utility functions.
\subsection{Utility and Preference Relations of the UEs and BSs}
%Substituting any utility function, $U_{m,n}$ that linearly depends on $r_{m,n}$ in \eqref{1a} yields the trivial solution that allocates the total bandwidth to the UE, i.e., $\alpha_{m,n}=1$. However, such a solution is not desirable since it does not consider any fairness. Moreover, the mmW-BS may allocate all the resources to a UE with low probability of LoS. With this in mind, 
For mmW links, a UE may experience multiple LoS/NLoS transmissions with different rates during the time that cell association is not updated. Thus, the utilities of UEs to BSs must be a function of the average rate. Here, we define the utility function of UE $m$ for BS $n$ as:
\begin{align}\label{util-UE}
U_{m}(n) &= \log\left[f(\boldsymbol{k}_{m,n})c^{\text{LoS}}_{m,n} +\left(1-f(\boldsymbol{k}_{m,n})\right)c^{\text{NLoS}}_{m,n}\right]\mathbbm{1}_{n \in \mathcal{N}_1}\notag\\
&+\log\left[c^{\mu\text{W}}_{m,n}\right]\mathbbm{1}_{n \in \mathcal{N}_2},
\end{align}
where, 
\begin{align}\label{beta}
\mathbbm{1}_{n \in \mathcal{N}_i} = 
\begin{cases}
1 &\text{if}\,\, n \in \mathcal{N}_i,\\
0, &\text{if}\,\, n \in \mathcal{N}_{j\neq i},
\end{cases}
\end{align}
and $\boldsymbol{k}_{m,n}$ is a vector composed of elements, $k_{m,n}(t')$ where $t'=t-1,t-2,\cdots, 0$, is the number of successful LoS transmissions from mmW-BS $n$ to UE $m$ and $f(k_{m,n}(t))$ is a metric that each UE uses to estimate the LoS probability $\rho_{m,n}$ for cell association at time $t$. In practice, the UEs can update a moving average of the number of LoS transmissions from each mmW-BS by using:
\begin{align}\label{f-func}
f(\boldsymbol{k}_{m,n}(t)) = \lambda \frac{k_{m,n}(t)x_{m,n}}{k}+(1-\lambda)f(\boldsymbol{k}_{m,n}(t-1)),
\end{align}
where $\lambda$ is a constant smoothing factor between $0$ and $1$ and $k$ is the number of transmission slots within the time window in which the association policy $\pi$ is not updated.
\begin{comment}
In addition, $k_{m,n}$ is the number of successful LoS transmissions from mmW-BS $n$ to UE $m$ and $f(k_{m,n})$ is a function of $k_{m,n}$.
\begin{proposition}
For a given matching $\pi$ and the utilities of the form $U_m(n)$ in \eqref{util-UE}, the optimal resource allocation for UEs $m \in \pi(n)$ of BS $n$ is
\begin{align}\label{RA}
\alpha_{m,n}(\pi) =  
\begin{cases}
\hat{\alpha} &\text{if}\,\, n \in \mathcal{N}_1,\\
\hat{\alpha}^\frac{1}{\log\left((1-\beta)f(k_{m,n})\right)}, &\text{if}\,\, n \in \mathcal{N}_2.
\end{cases}
\end{align}
where $\hat{\alpha}=\frac{1}{\card(\pi(n))}$.
\end{proposition}
\begin{proof}
See the Appendix.
\end{proof}
\end{comment} 
Using the utilities in \eqref{util-UE}, the preference relations of UEs are:
\begin{align}\label{prefer-UE}
n \succ_m n' \Leftrightarrow U_{m}(n) \geq U_{m}(n'), 
\end{align}
for $\forall m \in \mathcal{M}$, and $\forall n,n' \in \mathcal{N}$.

We note that assigning UEs to their most preferred BS may not admit a feasible matching in general. In other words, in order to satisfy the minimum quotas of the BSs, some UEs may have to be assigned to a lower ranked BS. Therefore, a suitable mechanism is required at the level of the BSs to determine which UEs must be assigned to the BSs with unsatisfied minimum quotas. To this end, all BSs must use the same preference profile, known as \textit{master list (ML)}, $\succ_{\textrm{ML}}$ with $\succ_n \,\equiv \,\succ_{\textrm{ML}}$, $\forall n \in \mathcal{N}$, as follow: 
\begin{align}\label{prefer-BS}
m \succ_{\textrm{ML}} m' \Leftrightarrow U_{\textrm{ML}}(m) \geq U_{\textrm{ML}}(m'),
\end{align}%\vspace{-.5cm}
where,  
\begin{align}\label{util-BS}
U_{\textrm{ML}}(m) = \{U_{m}(n')| U_{m}(n')\geq U_{m}(n), \forall n \in \mathcal{N} \}.
\end{align}
In fact, \eqref{util-BS} implies that BSs give higher priority to a UE that can achieve higher utility by being assigned to its preferred BS. This allows maximizing the sum utility in \eqref{1a}. To form the ML in practice, BSs only require to exchange the ordering of their nearby UEs to neighboring BSs.
\begin{algorithm}[t]
\footnotesize{
\caption{Proposed Cell Association Algorithm}\label{algo:1}
\textbf{Inputs:}\,\,$\succ_{\textrm{ML}}$, $\succ_m, \forall m \in \mathcal{M}$, $q_n^{\textrm{max}}$, $q_n^{\textrm{min}}, \forall n \in \mathcal{N}$.\\
\textbf{Outputs:}\,\, $\pi$, $\boldsymbol{x}$.
\begin{algorithmic}[1]
\State \text{Initialize:} $\pi(m)=\emptyset$, $\forall m \in \mathcal{M}$, $\mathcal{M}'=\mathcal{M}$.
\State Choose the UE $m^* \in \mathcal{M}'$ that has the highest rank in ML profile, i.e., $m^* \succ_{\textrm{ML}} m$, $\forall m \in \mathcal{M}'$.
\State Let $\pi(m^*)=n$, where $n$ is the most preferred BS based on $\succ_{m^*}$ with $\kappa_n<q_n^{\textrm{max}}$. Moreover, add $m^*$ to $\pi(n)$ and remove it from $\mathcal{M}'$.
\Repeat \,\,\,Steps 3 to 4 \Until{$\sum_{n \in \mathcal{N}} \lfloor q_n^{\textrm{min}}-\kappa_n\rfloor^{+}=|\mathcal{M}'|.$} 
\While{$\mathcal{M}' \neq \emptyset$}
\State Choose the UE $m^* \in \mathcal{M}'$ that has the highest rank in ML profile, i.e., $m^* \succ_{\textrm{ML}} m$, $\forall m \in \mathcal{M}'$.
\State Let $\pi(m^*)=n$, where $n$ is the most preferred BS based on $\succ_{m^*}$ with $\kappa_n<q_n^{\textrm{min}}$. Add $m^*$ to $\pi(n)$ and remove it from $\mathcal{M}'$.
\EndWhile 
\end{algorithmic}\label{Algorithm1}}
\end{algorithm}
\setlength{\textfloatsep}{0pt}% Remove \textfloatsep
\section{Proposed Cell Association and Load Balancing Algorithm}
\setlength{\textfloatsep}{0pt}
To solve the proposed cell association matching problem, we consider two important concepts of \textit{Pareto optimality} and \textit{two-sided stability}. 
A PO matching is defined as follow \cite{Fragiadakis12}:
\begin{definition}
A cell association policy, $\pi$, is \textit{Pareto optimal}, if there is no other feasible matching policy $\pi'$ such that $\pi'$ is preferred by all UEs over $\pi$, $\pi' \succeq_{m} \pi$, for all $m \in \mathcal{M}$, and strictly preferred over $\pi$, $\pi' \succ_{m} \pi$, for some UEs $m \in \mathcal{M}$. 
\end{definition}

In fact, PO is a widely adopted notion of efficiency for distributed mechanisms where each entity, here each UE, aims to maximize its own utility. Furthermore, the concept of two-sided \textit{stable matching} between UEs and BSs is defined as follows \cite{Roth92}:
\begin{definition}
A UE-BS pair $(m,n) \notin \pi$ is said to be a \textit{blocking pair} of the matching $\pi$, if and only if $m \succ_{n} m'$ for some $m' \in \pi(n)$ and $n \succ_m \pi(m)$.
Matching $\pi$ is \textit{stable}, if there is no blocking pair.
\end{definition}

A stable cell association policy ensures fairness for the UEs. That is, if a UE $m$ envies the assignment of another UE $m'$, then $m'$ must be preferred by the BS $\pi(m')$ to $m$, i.e., the envy of UE $m$ is not justified. When $\succ_n\, \equiv \, \succ_{\textrm{ML}}, \forall n \in\mathcal{N}$, as in our problem, the two-sided stable $\pi$ is also known as \textit{ML-fair} matching. 

For matching problems with no minimum quota, i.e., $q_n^{\textrm{min}}=0$, the well-known \emph{deferred acceptance (DA)} algorithm is used to find a stable matching such as in \cite{Roth92}, \cite{eduard11}, and \cite{6853635}. However, with minimum quotas, DA is no longer guaranteed to find a feasible solution.
\begin{proposition}
For cell association problems with minimum quota constraints, the standard DA algorithm may not admit a feasible solution.  
\end{proposition}
\begin{proof}
We prove this using an example. Let $\mathcal{M}=\{m_1,m_2,m_3\}$ and $\mathcal{N}=\{n_1,n_2,n_3\}$, with ML profile $m_1\succ_{\textrm{ML}}m_2\succ_{\textrm{ML}}m_3$. In addition, assume $q_n^{\textrm{min}}=1$, $q_n^{\textrm{max}}=2$ for all BSs, and $n_1\succ_{m_i}n_2 \succ_{m_i}n_3$, for all $m_i \in \mathcal{M}$. The DA algorithm for the UE-proposed solution yields $\pi(n_1)=\{m_1,m_2\}$, $\pi(n_2)=\{m_3\}$, and $\pi(n_3)=\emptyset$, which does not satisfy the minimum quota constraint for $n_3$.
\end{proof}

Therefore, a new algorithm must be developed to solve the problem. To this end, we propose the matching with minimum quota (MMQ) algorithm shown in Algorithm \ref{Algorithm1}, which is designed based on \cite{Fragiadakis12}. The proposed algorithm proceeds as follows. After initialization, in step 2, UE $m^*$ with the highest rank in the ML profile requests a connection with its most preferred BS $n$. If $\kappa_n$ is less than its maximum quota, UE $m^*$ will be accepted by BS $n$. This procedure continues in Steps 3 and 4 for the remaining UEs until the number of UEs is equal to the required number of UEs for meeting the minimum quota constraints, i.e., $\sum_{n \in \mathcal{N}} \lfloor q_n^{\textrm{min}}-\kappa_n\rfloor^{+}=|\mathcal{M}'|$, where $\lfloor x \rfloor^+=max(x,0)$. Next, in Step 7, the most preferred UE based on the ML profile must be assigned only to its most preferred BS from the subset of $\mathcal{N}$ with $\kappa_n<q_n^{\textrm{min}}$. In fact, our algorithm allows each UE to be assigned to its most preferred BS, as long as the minimum and maximum quota constraints are not violated. The algorithm terminates once all the UEs are assigned to a BS. The proposed, distributed matching algorithm exhibits the following properties:
\begin{theorem}\label{theorem1}
Algorithm \ref{Algorithm1} is guaranteed to yield a feasible PO and stable matching between UEs and BSs.
\end{theorem}
\begin{proof}
If the cell association $\pi$, given by Algorithm \ref{Algorithm1} is not PO, a UE $m$ must exist that can benefit by being assigned to another BS $n$, i.e., $n\succ_m \pi(m)$. There are two possible cases to consider. First, $n\succ_m \pi(m)$ and $m \notin \pi(n)$ imply that UE $m$ has applied to BS $n$ prior to $\pi(m)$ and is rejected, due to $\kappa_n=q_{n}^{\textrm{max}}$ and $m' \succ_{\textrm{ML}} m$, for all $m' \in \pi(n)$. Therefore, adding $m$ to $\pi(n)$ does not yield a feasible solution. Second, UE $m$ is assigned to $\pi(m)$ to satisfy minimum quota constrain for $\pi(m)$. This means re-allocating $m$ to BS $n$ will violate the minimum quota criterion for $\pi(m)$ and is not feasible. Therefore, the given solution is feasible Pareto optimal.

To prove the stability, we note that if UE $m$ prefers to be assigned to BS $\pi(m')$, that implies $m' \succ_{\textrm{ML}} m$, otherwise, $\pi(m)=\pi(m')$. Hence, no blocking pair exists and the solution is stable.    
\end{proof}

We must note that Pareto optimality and stability cannot be inherently achieved if the BSs do not follow the ML preference profile. In fact, for $\succ_n \neq \succ_{\textrm{ML}}$, there is no algorithm in general that can guarantee a feasible PO and stable solution \cite{Fragiadakis12}. 
%%%%%%%%%%%%%%%%%%%%%%%%%%%%%%%%%%%%%%%%%%%%%%%%%%%%
\vspace{-.1cm}
\section{Simulation Results}\label{sec:sim}\vspace{-.1cm}
For simulations, we consider a network with $N_1=10$ mmW-BSs, $N_2=10$ $\mu$W-BSs, and up to $M=100$ UEs located uniformly and randomly over an area with diameter $r=1$ km. The main parameters are summarized in Table \ref{tabsim}. The average probability of LoS for each mmW BS-UE pair is sampled from a uniform distribution, $\rho_{m,n} \in \left[0,1\right]$. All statistical results are averaged over a large number of independent runs.

We compare the performance of the proposed MMQ algorithm with both conventional max-SINR and max-RSSI approaches. We also consider a CRE with bias factor $\gamma_{\text{RSSI}}$ and $\gamma_{\text{SINR}}$, respectively, for the max-RSSI and max-SINR schemes for further comparisons. To calculate the rates, the total bandwidth at each BS is allocated equally to the associated UEs. That is, $r^{\text{mmW}}_{m,n}=\frac{w_1}{\kappa_n}c_{m,n}^{\text{mmW}}$, where $r^{\text{mmW}}_{m,n}$ denotes the achievable rate for UE $m$ associated with mmW-BS $n$. Moreover, $r^{\mu\text{W}}_{m,n}=\frac{w_2}{\kappa_n}c_{m,n}^{\mu\text{W}}$, where $r^{\mu\text{W}}_{m,n}$ denotes the achievable rate for UE $m$ assigned to $\mu$W-BS $n$.
In \cite{6497017}, it is shown that for logarithmic utilities, as in \eqref{util-UE}, uniform resource allocation maximizes the sum utility. 
{\fontsize{7}{20}
	\begin{table}[!t]
	
	\centering
	\caption{%\mycaption{%\vspace*{-1em}
		\vspace*{-0em}Simulation parameters}\vspace*{-0em}
	\begin{tabular}{|c|c|c|}
	\hline
	\bf{Notation} & \bf{Parameter} & \bf{Value} \\
	\hline
	$p_n$ & Transmit power & $30$ dBm\\
	\hline
	$(\omega_1,\omega_2)$ & Bandwidth & ($1$ GHz, $10$ MHz)\\
	\hline
	($\xi_{1,\text{LoS}},\xi_{1,\text{NLoS}}, \xi_2$) & Standard deviation of path loss& ($5.2,7.6,10$) \cite{andrew16} \\
	\hline
	($a_{1,\text{LoS}},a_{1,\text{NLoS}},a_{2}$) & Path loss exponent& (2,4,3) \cite{andrew16}\\
	\hline
	($b_1, b_2$) & Path loss at $1$ m& ($70,38$) dB\\
	\hline
	$\psi$ & Antenna gain& $18$ dBi \cite{andrew16}\\
	\hline
	$N_0$ & Noise power spectral density& $-174$ dBm/Hz \\
	\hline
	$M$ & Number of UEs& From $10$ to $100$\\
	\hline
	$q_n^{\text{max}}$ & Maximum quota& $M$\\
	\hline
	\end{tabular}\label{tabsim}\vspace{-0cm}
	\end{table}}

Fig. \ref{sim1} shows the average sum-rate for the proposed MMQ approach, compared to max-RSSI and max-SINR approaches versus the number of UEs. The bias factors are chosen such that near uniform loads are achieved for all the BSs. The minimum quotas for $\mu$W-BSs are chosen randomly from $0$ to $\lfloor M/N_2\rfloor$, with $\lfloor . \rfloor$ denoting the floor operand. The results show that the proposed approach achieves up to $14\%$ and $18\%$ improvements compared to, respectively, the max-SINR and the max-RSSI schemes, for $M=50$. This is due to the fact that the achievable rate is a nonlinear function of the SINR or RSSI metrics. Hence, average SINR or RSSI , with respect to $\zeta_{m,n}$, cannot be used to find the average achievable rate. However, the proposed approach directly relies on the average achievable rate, as shown in \eqref{util-UE}. 

Fig. \ref{sim7} shows the optimal minimum quota for $\mu$W-BSs that yields the maximum average sum-rate, as the number of UEs varies, for different values of $N_1=N_2$. The minimum quota for mmW-BSs is zero, since the load of the mmW-BSs are higher than $\mu$W-BSs. The results show that the optimal minimum quota increases, as $M$ increases, since more UEs must be associated with the $\mu$W-BSs. Moreover, the optimal $q_n^{\textrm{min}}$ decreases as $N_1$ and $N_2$ increase, since more BSs are available. For $M=100$ UEs, we observe that the optimal minimum quotas are $q_n^{\textrm{min}}=8$, for all $n \in \mathcal{N}_2$, which implies that $80 \%$ of the UEs must be assigned to the $\mu$W-BSs. Hence, if sum rate is considered as the optimality criterion, the result does not yield a balanced network. 

\begin{figure}[!t]
\centering
\centerline{\includegraphics[width=8.4cm]{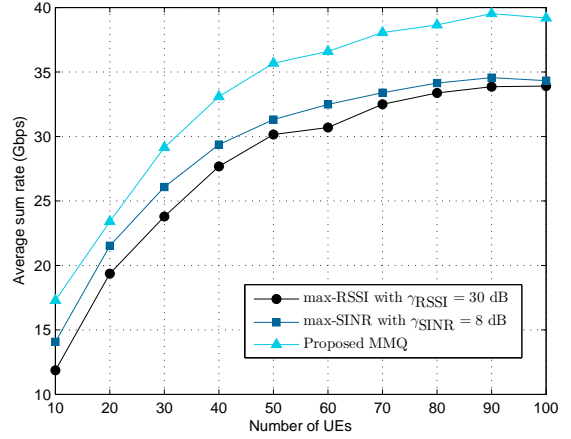}}\vspace{-0.3cm}
\caption{The average sum-rate (Gbps) versus the number of UEs $M$.}\vspace{-0.5cm}
\label{sim1}
\end{figure}

\begin{figure}[t!]
\centering
\centerline{\includegraphics[width=8.4cm]{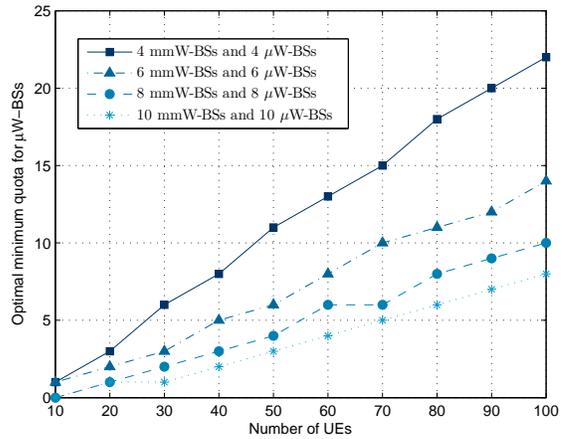}}\vspace{-0.3cm}
\caption{The optimal quota values for $\mu$W-BSs versus the number of UEs.}\vspace{-.5cm}
\label{sim7}
\end{figure} 

\begin{figure}[t!]
\centering
\centerline{\includegraphics[width=8.4cm]{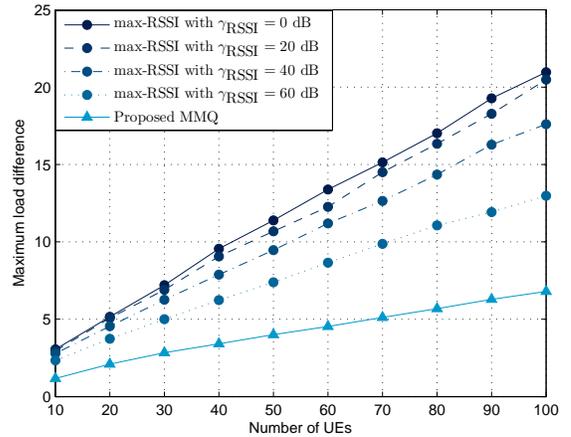}}\vspace{-0.3cm}
\caption{The maximum load difference, $\Delta_{\kappa}$, for the proposed MMQ approach, compared to the max-RSSI with CRE.}\vspace{0.2cm}
\label{sim6}
\end{figure}

In Fig. \ref{sim6}, the maximum load difference $\Delta_{\kappa}$, is evaluated for the proposed algorithm compared to max-RSSI approach with CRE under biasing values ranging from $0$ to $60$ dB. The results show that, as biasing increases, the load balancing decreases and then increases. For all biasing values, $\Delta_{\kappa}$ for the max-RSSI approach is significantly larger than the proposed MMQ algorithm. In fact, we observe that the proposed MMQ algorithm substantially improves the load balancing, reaching up to $48 \%$ compared to the max-RSSI with $\gamma_{\text{RSSI}}=40$ dB for $M=70$. This improvement is due to the fact that the CRE with biasing cannot precisely control the number of UEs re-associated from $\mu$W-BSs to the mmW-BSs. However, in the proposed approach, the BSs can directly control the number of associated UEs by adjusting their minimum quotas.  
\begin{figure}
\centering
\centerline{\includegraphics[width=8.4cm]{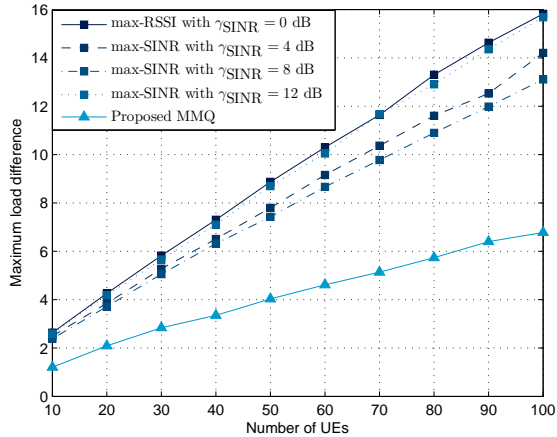}}\vspace{-0.3cm}
\caption{The maximum load difference, $\Delta_{\kappa}$, for the proposed MMQ approach, compared to the max-SINR with CRE.}\vspace{-.5cm}
\label{simulation7}
\end{figure}

Fig. \ref{simulation7} compares the maximum load difference resulting from the proposed MMQ algorithm, compared to the max-SINR approach with CRE. We observe that, as $\gamma_{\text{SINR}}$ is increased from $0$ to $8$ dB, the maximum load difference decreases. However, for $\gamma_{\text{SINR}}>8$ dB, the load difference increases, since a larger number of UEs is being assigned to the $\mu$W-BSs. Moreover, Fig. \ref{simulation7} shows that for all network sizes, the proposed approach substantially outperforms the max-SINR approach with CRE. In fact, the proposed approach decreases $\Delta_k$ by $47\%$, compared to max-SINR with $\gamma_{\text{SINR}}=8$ dB for $M=70$. Here, we can once again see that the minimum quota constraints allow BSs to control the load more precisely, compared to max-SINR with CRE.
\begin{figure}
\centering
\centerline{\includegraphics[width=8.4cm]{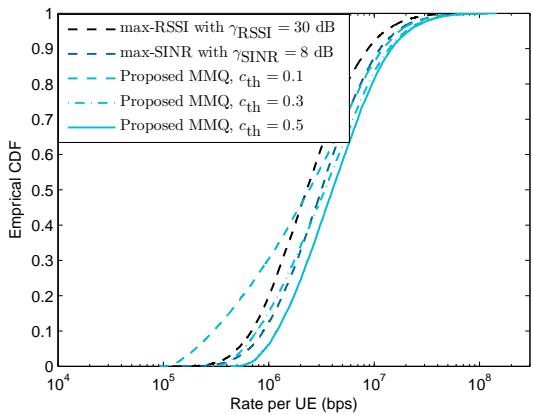}}\vspace{-0.3cm}
\caption{The empirical CDF of the rate per UE over $\mu$W frequency band for $M=100$ UEs.}\vspace{0.3cm}
\label{sim4}
\end{figure}

In Fig. \ref{sim4}, the statistics of the average rate per UE are shown over the $\mu$W band, compared to max-RSSI and max-SINR. Here, we focus on the average rate for $\mu$W links, since the mmW links achieve higher rates, due to the available bandwidth. The results show that, an inherent byproduct of any load balancing technique is the fact that some of the UEs will eventually be associated with an unpreferred $\mu$W-BS to satisfy the minimum quota constraints. Such UEs will then trade off rate for load balancing. To this end, parameter $c_{\text{th}}$ is defined as a utility threshold for UEs. That is, the UE $m$ is assigned to an unpreferred $\mu$W-BS $n$, if $U_{m}(n)\geq c_{\text{th}}$. $c_{\text{th}}$ allows controlling the tradeoff between a highly balanced load and a low average rate for the cell edge UEs. Fig. \ref{sim4} shows that for $c_{\text{th}}=0.5$, the proposed MMQ algorithm outperforms the max-RSSI and the max-SINR approaches with CRE.

%The effect of load balancing on the network load distribution is shown in Fig. \ref{sim6}. Compared to the PL-based and SINR-based cell associations, the proposed MMQ algorithm can yeild an even load distribution over the mmW and $\mu$W frequency bands. Considering the results in Figs. \ref{sim5} and \ref{sim6}, it can be observed that the proposed approach provides a substantial gains in terms of the rate, while maintaining a balanced network.
\vspace{-0.1cm}
\section{Conclusions}
In this paper, we have proposed a novel cell association and load balancing framework for small base stations operating at mmW and $\mu$W frequency bands. We have formulated the problem as a one-to-many matching game with minimum quotas. To solve this game, we have proposed a distributed algorithm that considers the average LoS probability in addition to the achievable rate, while assigning UEs to the BSs. We have shown that the proposed algorithm yields a Pareto optimal and stable association policy. Simulation results have shown that the proposed MMQ algorithm outperforms the conventional max-RSSI and max-SINR schemes in terms of both performance and load balancing.
\vspace{-0.1cm}

\bibliographystyle{IEEEbib}
\bibliography{references}
\end{document}